\documentclass[10pt]{article}
\usepackage{algorithm2e}
\usepackage{algorithmicx}
\usepackage{epsfig}
\usepackage{fancybox}
\usepackage{graphics}
\usepackage{amssymb}
\usepackage{geometry}
\usepackage{fancyhdr}
\usepackage{esvect}

\newtheorem{definition}{\textbf{Definition}}[section]

\newtheorem{theorem}{\textbf{Theorem}}[section]

\title{Planar $\beta$-skeletons via point location in monotone subdivisions
of subset of lunes
\thanks{This research is supported by the ESF EUROCORES programme EUROGIGA, CRP VORONOI.}}
\author{
        Miroslaw Kowaluk 
        \thanks{Institute of Informatics, University of Warsaw, Poland {\tt kowaluk@mimuw.edu.pl}}
}

\begin{document}

\maketitle

\begin{abstract}
We present a  new algorithm for lune-based $\beta$-skeletons for sets of $n$ points in the plane,
for  $\beta \in (2,\infty]$,  the only case when optimal algorithms are not known. 
The running time  of the algorithm is $O(n^{3/2} \log^{1/2} n)$, which is the best known
and is an improvement of Rao and Mukhopadhyay \cite{rm97} result.
The method is based on point location in monotonic subdivisions of arrangements of curve segments. 
\end{abstract}

\section{Introduction}




$\beta$-skeletons \cite{kr85} belong to the family of
proximity graphs, geometric graphs in which two vertices (points) produce an edge if and only if 
they satisfy particular geometric requirements. 
The proximity graphs are both important and popular because of many practical applications. 
They span a broad spectrum of areas, from earlier work on mathematical 
morphology on lattices and graphs, pattern recognition, geographic information systems, 
data mining to more recent applications in ad hoc wireless networks and  systems biology. 
Requirements defining proximity graphs can be formulated in many metrics, with the Euclidean metric 
being one of the most commonly used. 

\begin{definition}
For a given set $P$ containing $n$ points in $R^2$, distance function $d$ and the 
parameter $\beta$  we define $\beta$-skeleton as a graph $(P , E)$, in which $xy \in E$ 
iff no point in $P \setminus \{x,y\}$ belongs to $R(x,y, \beta)$, where 
\begin{itemize}
\item 
for $\beta = 0$, $R(x,y, \beta)$ is the segment $xy$;
\item 
for $0 < \beta < 1$, $R(x,y, \beta)$ is the intersection  of two discs with the radius 
$d(x,y)/2\beta$, which boundaries contain the both points $x$ and $y$. 
\end{itemize}
For  $1 \leq \beta \leq \infty$ there are two ways to define the region $R(x,y, \beta)$ leading to two different 
families of graphs. \\
The lune-based definition is as follows:
\begin{itemize}
\item 
For  $1 \leq \beta \leq \infty$ , $R(x,y, \beta)$ is a intersection of two discs with radius 
$\beta d(x,y)/2$ and centered in points $(1-\beta/2)x+(\beta/2)y$ and $(\beta/2)x+(1-\beta/2)y$, 
respectively.   
\item 
For $\beta = \infty$, $R(x,y, \beta)$ is an unbounded strip between two lines containing $x$ 
and $y$, respectively, and perpendicular to the segment $xy$. 
\end{itemize} 
The second definition defining a different family of graphs is circle-based.
\begin{itemize}
\item
For $1 \leq \beta < \infty$, $R(x,y, \beta)$ is the union  of two discs, each having diameter 
$\beta d(x,y)$ and having the segment $xy$ as its chord.
 \item
For $\beta = \infty$, $R(x,y, \beta)$ is the union of hyperplanes having the segment $xy$ on 
their border.
\end{itemize}

\end{definition}

\begin{figure}[hbt]
\begin{center}
\includegraphics[scale=0.3]{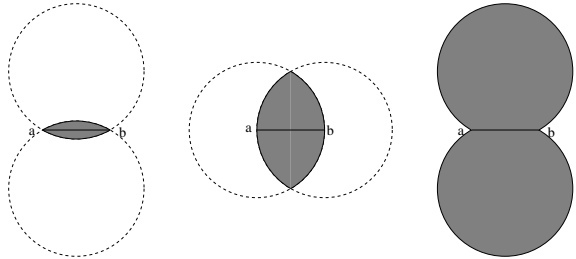}
\caption{ The region $R(a,b, \beta)$ for $ 0 < \beta < 1$ (left), the lune-based region 
$R(a,b, \beta)$ for $1 \leq \beta < \infty$ (middle) and the circle-based region $R(a,b, \beta)$
for $1 \leq \beta < \infty$ (right).  }
\label{fig:1}
\end{center}
\end{figure}

The open (closed, respectively) area $R(x,y, \beta)$ defines an open (closed, respectively) $\beta$-skeleton. 
A Gabriel Graph (GG) \cite{gs69} is a closed $1$-skeleton and a Relative Neighborhood Graph (RNG) 
\cite{t80}  is an open $2$-skeleton. \\


For $0 \leq \beta < 1$, there are in the worst case optimal $O(n^2)$ time algorithms \cite{hlm03} 
for both lune-based and circle-based 
$\beta$-skeletons.  Also, any $\beta$-skeleton for $1 \leq \beta \leq 2$ 
can be computed in the optimal $O(n \log n)$ time (\cite{jky89,l94,ms80,s83}). 
For circle-based 
$\beta$-skeletons, where $1 \leq \beta \leq \infty$, there also is an optimal $O(n \log n)$ time 
algorithm \cite{kr85}. 

For $\beta > 2$, the lune-based and circle-based regions $R(x, y, \beta)$ are
very different and it has been reflected by higher running times
of the best known algorithms for the lune-based $\beta$-skeletons.
In particular, for $2 < \beta \leq \infty$, the fastest algorithm that 
has been reported for computing 
lune-based $\beta$-skeletons requires $O(n^{3/2} \log {n})$ time \cite{rm97}.
The algorithm uses  the sweeping-line technique. \\
In this paper, we will show a $O(n^{3/2} \log^{1/2} n)$ 
time algorithm computing a lune-based 
$\beta$-skeleton for  $2 < \beta \leq \infty$. To this end, we will use 
arrangements of curve segments \cite{agr00} and data structure for point location
in monotone subdivisions \cite{egs86}. 

\section{Algorithm}

It is well known, see e.g. \cite{kr85}, that $\beta$-skeletons for $2 < \beta$ are subgraphs 
of the Delaunay triangulation of the input point sets. For a set $P$ of $n$ points,  they have size $O(n)$. 
Thus, to construct a $\beta$-skeleton, we can use their definition to eliminate edges that do not belong 
to $\beta$-skeleton by comparing edges of Delaunay triangulation against all of the points in the input set. 
To this end, we analyze lunes (for $2 < \beta < \infty$) 
or strips (for $\beta = \infty$) defined by the edges of the DT triangulation. \\
Specifically, we have to identify these lunes or strips that contain points
in $P$. \\
Let us consider $m$ edges of the Delaunay triangulation of $P$. 
$m$ will be determined later to minimize the running time. \\

\begin{figure}[hbt]
\begin{center}
\includegraphics[scale=0.4]{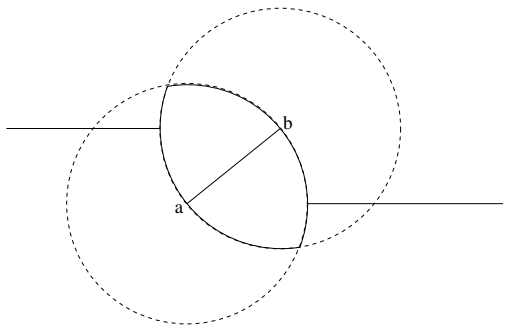}
\caption{ A subdivision of the plane into three monotone regions for one lune. }
\label{fig:2}
\end{center}
\end{figure}

We construct arrangement of lunes defined by these edges. 
Specifically, for each of the lunes, we draw two horizontal half-lines 
that start at the points with the smallest and, respectively, largest 
$x$-coordinates. They yield three monotone regions per lune (see Fig. \ref{fig:2}).    
The size of the arrangement is $O(m^2)$ 
and the time complexity of its construction is the same \cite{agr00}. \\
The intersection of monotone regions is monotone. The above construction 
results in a monotone subdivision of size $O(m^2)$ for the arrangement of size $O(m^2)$. \\
Then, we locate points of $P$ in regions of the subdivision. The construction time of an auxiliary data structure 
is linear with respect to a size of the subdivision \cite{egs86}.  Hence it can be done in $O(m^2)$ time and $O(m^2)$ 
space. A query time is $O(\log m)$. The results of querying are stored in respective cells of the subdivision 
data structure. \\
We traverse the dual graph of the subdivision in a Depth-First-Search fashion
in order to identify (and store) edges of the Delaunay triangulations that do not belong the $\beta$-skeleton 
and lunes (strips, respectively)
containing currently visited region of the subdivision. \\

\begin{figure}[hbt]
\begin{center}
\includegraphics[scale=0.3]{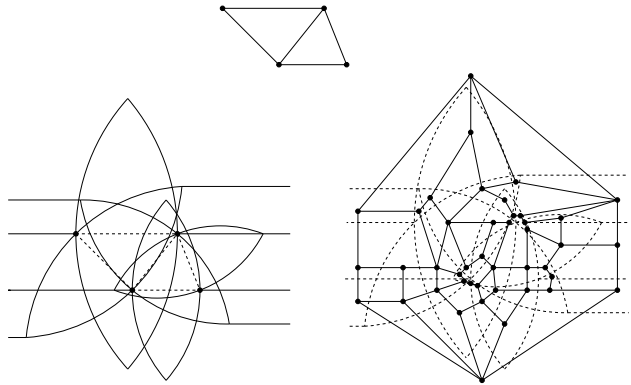}
\caption{ The Delaunay triangulation for a set of four points, the corresponding 
subdivision and its dual graph. }
\label{fig:3}
\end{center}
\end{figure}

The structure is a table of lunes (strips, respectively) forming the subdivision with 
\begin{itemize}
\item
records storing - value {\em true} when the respective lune contains some point of $P$ and {\em false} 
otherwise and
\item
pointers to a double linked list of lunes (strips, respectively) containing currently visited region which 
respective records are {\em false}.
\end{itemize}
Moving between two neighboring regions we visit or leave the lune or our status stays unchanged. \\
In the first case : 
\begin{itemize}
\item
if record of the auxiliary structure of this lune has value {\em false} we add the lune to the list 
in the structure; 
otherwise the list stays unchanged, 
\item 
if visited region contains a point of $P$ we change values of records of all lunes belonging to the list 
on {\em true} and erase the list.
\end{itemize}
In the second case we remove the lune form the list if its record has value {\em false}. \\
In the third case all data structures stay unchanged. \\
Visiting all regions of the subdivision take $O(m^2)$ time. \\
We repeat the above process for all the groups in the partition of the edges into
groups of size $m$.

\begin{theorem}
For $2 < \beta \leq \infty$, $\beta$-skeleton of a set $P$ of $n$ points in the plane can be constructed 
in $O(n^{3/2} \log^{1/2} n)$ time and $O(n \log n)$ space.
\end{theorem}
\begin{proof}
The time complexity of every steps of the algorithm is 
$O(n \log n) + n/m(O(m^2) + O(m^2) + O(n \log m) + O(m^2)) =  O(n \log n) + n/m(O(m^2) + O(n \log m))$. 
The value of this function is minimal for $m = (n \log n)^{1/2}$ . Hence the time complexity of the algorithm 
is $O(n^{3/2} \log^{1/2} n)$. Space complexity of the algorithm is $O(n) + O(m^2) = O(n \log n)$ .  
\end{proof}


\section{Conclusion}


We have presented a new algorithm for the lune based $\beta$-skeletons, where $\beta > 2$. For circle-based skeletons the approach would work as well, 
however better algorithms are known in that case.
Although the algorithm provides only an incremental improvement
compared to the previous one, the improvement is interesting due to
a tantalizing jump in the algorithmic complexity of the $\beta$-skeleton problem
that occurs for $\beta = 2$. While for $\beta \leq 2$ optimal
algorithms are known,  the optimality is open for $\beta > 2$.
The presented algorithm narrows the gap, however more research is needed to
better understand the reasons for this apparent difficulty for $\beta > 2$.

Finally, we believe that the point location approach can be also used for
metrics other than Eucliedean. 




\small 
\bibliographystyle{abbrv}

\end{document}